\documentclass[letterpaper, 10 pt, conference]{ieeeconf} 

\IEEEoverridecommandlockouts                              

\usepackage{amsmath}
\usepackage{amsfonts}
\usepackage{amssymb}
\usepackage{upgreek}
\usepackage{graphicx}
\usepackage{float}

\def\usetikz{0} 

\if\usetikz1
	\usepackage{tikz}
	\usepackage{circuitikz} 
	\usepackage{pgfplots}
	\pgfplotsset{compat=1.5.1,
		small/.style={
			tick label style={font=\scriptsize},
			label style={font=\small},
			height=2.75cm,
			width=3.75cm
		}
	}
	\usepgfplotslibrary{groupplots}
	\usetikzlibrary{positioning}
	\usetikzlibrary{external} 
	\tikzset{terminal/.style=[node distance=1.25cm,
		fblock/.style={rectangle,minimum size=1.0cm,
		thick,draw=black},
		sum/.style={circle,minimum size=0.5cm,
		thick,draw=black},
		every new ->/.style={-latex},
		every new --/.style={},
		hvpath/.style={to path={-| (\tikztotarget)}},
		vhpath/.style={to path={|- (\tikztotarget)}}
	}
\else
	\usepackage{color}
	\graphicspath{{./tikz/}}
\fi

\usepackage{amsthm}
\newtheorem{lem}{Lemma}
\newtheorem{thm}{Theorem} 

\theoremstyle{definition}

\newtheorem{defn}{Definition}

\newcommand{\setreal}{\mathbb{R}}
\newcommand{\setint}{\mathbb{Z}}
\newcommand{\setnat}{\mathbb{N}}

\newcommand{\Fint}{F_{\text{int}}}

\newcommand{\nlayers}{L}
\newcommand{\npoles}{n}

\newcommand{\nunits}{M}

\newcommand{\iApp}{i}

\newcommand{\itdnn}{n}
\newcommand{\igobf}{m}
\newcommand{\spm}{\hat{k}_s}
\newcommand{\spd}{k_s}

\newcommand{\samplingT}{t_s}

\newcommand{\inpbound}{\beta}

\newcommand{\inpclass}{\mathcal{U}_\inpbound}

\title{\LARGE \bf
System identification of biophysical neuronal models*
}

\author{Thiago B. Burghi$^{1}$, Maarten Schoukens$^{2}$
and Rodolphe Sepulchre$^{1}$
\thanks{* Slightly extended pre-print of the paper to be presented 
at the 59th Conference on Decision and Control, held remotely
between December 14-18, 2020.
The research leading to these results has received
funding from the Coordena\c{c}\~{a}o de Aperfei\c{c}oamento 
de Pessoal de N\'{i}vel Superior (CAPES) -- Brasil 
(Finance Code 001) and the 
 European Research Council (under the
Advanced ERC Grant Agreement Switchlet n.670645).}
\thanks{$^{1}$Thiago B. Burghi ({\tt\small tbb29@cam.ac.uk})
and Rodolphe Sepulchre 
({\tt\small r.sepulchre@eng.cam.ac.uk}) are with
the Department of Engineering, Control Group, 
University of Cambridge, CB2 1PZ, UK.}%
\thanks{$^{2}$Maarten Schoukens ({\tt\small m.schoukens@tue.nl}) is with the Department of Electrical
Engineering, Eindhoven University of Technology, 5612 AZ
Eindhoven, Netherlands.
}%
}

\begin{document}

\maketitle
\thispagestyle{empty}
\pagestyle{empty}



\begin{abstract}         
After sixty years of quantitative biophysical modeling 
of neurons, the identification of neuronal dynamics 
from input-output data remains a challenging problem, 
primarily due to the inherently nonlinear nature of excitable behaviors.
By reformulating the problem in terms of the identification
of an operator with fading memory, we explore a simple approach 
based on a parametrization given by a series interconnection 
of Generalized Orthonormal Basis Functions (GOBFs) and 
static Artificial Neural Networks. We show that GOBFs are 
particularly well-suited to tackle the identification 
problem, and provide a heuristic for selecting GOBF poles 
which addresses the ultra-sensitivity of neuronal behaviors. 
The method is illustrated on the identification of a 
bursting model from the crab stomatogastric ganglion.
\end{abstract}
          

\section{Introduction}

%

This  paper explores the potential of a simple
parametrization for the identification of a nonlinear
system that can be represented as the feedback 
interconnection between an integrator and an operator
with fading memory \cite{boyd_fading_1985}, as shown in Figure 
\ref{fig:passive_fm_feedback}. 
Our interest in this particular structure is that it
encompasses most biophysical models of neuronal circuits
\cite{burghi_feedback_2020,burghi_feedback_2021}. In such 
models, the integrator represents the neuronal membrane model, 
whereas the operator with fading memory represents the 
input-output mean-field relationship between voltage 
and the internal currents arising from the opening 
and closing of myriad ion channels. Feedback between
these two components destroys the fading memory property, and indeed this 
property is ruled out by observed neuronal behaviors 
such as excitability, autonomous oscillations, and chaos. 

The motivation to preserve the biophysical decomposition
of Figure \ref{fig:passive_fm_feedback} 
in a system identification setting is obvious: the
identification of nonlinear systems with fading memory
is a mature topic
\cite{ljung_convergence_1978,schram_system_1996,schoukens_identification_2017},
whereas the identification of feedback systems lacking this
property is challenging. In the latter case, it is difficult to
obtain any guarantees on the behavior of the identified system. 
While the availability of noiseless state measurements allows
some results to be established \cite{manchester_identification_2011-1},
internal neuronal states cannot be measured. This 
limitation poses further issues concerning model identifiability 
\cite{wigren_nonlinear_2015} 
and optimization tractability
\cite{abarbanel_dynamical_2009}. Furthermore, the lack
of fading memory prohibits the use of contraction constraints
recently advocated to improve the trainability of recurrent
models \cite{revay_contracting_2020}.
To avoid some of these difficulties, we explore directly 
identifying the nonlinear fading memory component using a 
feedforward parametrizaton with universal approximation properties: 
the series interconnection of Generalized Orthonormal Basis 
Functions (GOBFs) with a static Artificial Neural Network 
(ANN) nonlinearity.

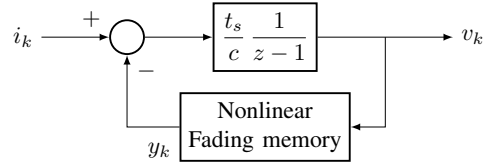
\begin{figure}[t]
	\centering
	\if\usetikz1
	\begin{tikzpicture}
		\node (v_dyn) at (0,0) [fblock,align=center]
			{$\displaystyle\frac{\samplingT}{c}
			\,\displaystyle\frac{1}{z-1}$};		
		\node (int_dyn) [fblock,
		label={185:$y_k$},
		below= 0.3 of v_dyn,
		align=center]{Nonlinear \\ Fading memory};
		\node (sum1) [sum,left=of v_dyn,
				label={175:$+$},label={-85:$-$}] {};		
		\coordinate[left= 1 of sum1] (input);
		\coordinate[above= .5 of sum1] (noise);
		\coordinate[right=of v_dyn] (split1);	
		\coordinate[right=of split1] (output);		
		\draw[-latex] (sum1) -- (v_dyn);
		\draw[-latex] (input) -- (sum1) 
			node[left,pos=0]{$i_k$};
		\draw[-latex] (v_dyn) -- (split1) -- (output) 
		node[right,pos=1]{$v_k$};
		\draw[-latex] (split1) |- (int_dyn);
		\draw[-latex] (int_dyn) -| (sum1);
	\end{tikzpicture}
	\else
		\includegraphics[scale=0.9]{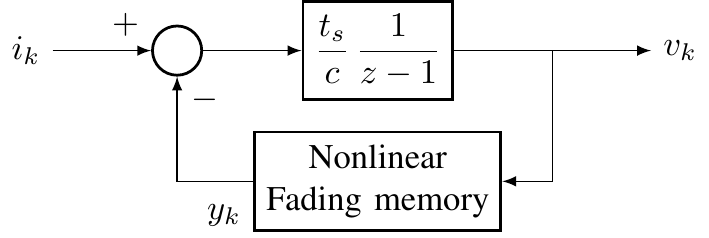}
	\fi
	\caption{The general model structure of a neuronal
			system.}
	\label{fig:passive_fm_feedback}
\end{figure}

The contributions of the present paper are
the following: we show how the problem of neuronal system
identification can be cast in terms of a problem of
identifying a fading memory operator; we present a proof
that the series interconnection of GOBFs with ANNs implements 
a universal approximator; and we provide a heuristic for choosing 
GOBF poles to identify neuronal systems, based on the input-output
property of ultrasensitivity
\cite{drion_neuronal_2015,franci_sensitivity_2019}.
%



The paper is structured in the following way:
In Section \ref{sec:problem}, we formulate the
identification problem and argue that 
the model class
of Figure \ref{fig:passive_fm_feedback} encompasses
most neuronal systems. In Section \ref{sec:identification}
we prove that the proposed model structure is a universal 
approximator of fading memory operators and
discuss how to choose the poles of the GOBFs for
neuronal systems. In Section
\ref{sec:example}, we illustrate how the model
structure can be used to identify a bursting neuronal
model. 



\section{Problem statement}
\label{sec:problem}

In this section, we introduce the general problem of
neuronal system identification. We restrict the treatment
in this paper to single-input single-output systems.
We write $\setint_+ = \{0,1,\dotsc\}$ and 
$\setnat = \{1,2,\dotsc\}$. The set 
$\setreal(\setint_+)$ is the set of real-valued sequences
defined on $\setint_+$, and $\ell_\infty(\setint_+)$ is the
set of all real-valued bounded sequences defined on 
$\setint_+$. 

\subsection{Neuronal systems}

A neuronal system dictates the evolution of the
cellular membrane potential of a neuron, denoted by
$v\in\setreal$. The voltage $v$ changes in time due
to the flow of ionic currents through the membrane,
which can be reasonably represented by a capacitor.
Thus any discrete-time single-compartment
biophysical neuronal system has the form
\begin{equation}
			\label{eq:DT_v_system}
				c\frac{v_{k+1}-v_k}{\samplingT} =  
				-y_k + \iApp_k 
\end{equation}
where $c>0$ is the membrane capacitance, $\samplingT>0$ is 
a sampling period, $i_k\in\setreal$ is an external applied
current, and 
\begin{equation}
	\label{eq:fm_channels} 
	y_k = (\Fint v)_k
\end{equation}
represents the aggregate internal ionic currents traversing
the neuronal membrane (we assume forward-Euler 
discretization of the physical system). 
The time-invariant 
operator $\Fint:\ell_\infty(\setint_+)\to\setreal(\setint_+)$
represents the dynamics of ion channels embedded in the
membrane, and it is the accurate modeling of this operator
which allows neuronal behavior to be successfully
reproduced in neurocomputational studies.

\subsection{The internal current operator has fading
memory} 

The operator $\Fint$ is a mean-field model of complex
molecular events dictating the opening and closing of ion
channels. Ever since the pioneering work of Hodgkin and
Huxley \cite{hodgkin_quantitative_1952}, ion channel
kinetics have been primarily modeled using step response
experiments  
implemented by means of the so-called
\textit{voltage-clamp} technique.
It is the design of such experiments that led to the foundation 
of biophysical modeling in neurophysiology. 
In an experimental setting, the voltage output $v_k$ is
measured, and the experimenter has full control over the 
current input $i_k$. The voltage-clamp experiment implements a
high-gain output feedback law
\begin{equation}
	\label{eq:output_feedback} 
	i_k = \gamma(v_k - r_k)
\end{equation}
where $r_k$ is a reference signal chosen by the
experimenter, and $\gamma>0$ is the feedback gain.
When the voltage-clamp experiment is carried out in a 
biological neuron, it is observed that for any sufficiently high gain
$\gamma>0$, the measured $v_k$ closely tracks the reference 
$r_k$; furthermore, the tracking error decreases as $\gamma$ is 
increased. This is in stark contrast with the open-loop behavior
of the neuron, as illustrated in Figure \ref{fig:HH_full_Example}.
The behavior observed experimentally confirms the 
basic capacitive modeling assumption behind
\eqref{eq:DT_v_system}: it is the behavior of a system with
relative degree one, whose inverse (the system with input 
$v_k$ and output $i_k$) can be obtained at the limit of
high gain $\gamma \to \infty$. 

Furthermore, the voltage-clamp experiment indicates that
the system inverse, and thus also $\Fint$, have fading
memory, in the following sense:
\begin{defn}[\cite{park_criteria_1992,boyd_fading_1985}]
For arbitrary $\inpbound>0$, 
consider the class $\inpclass\subset \ell_\infty(\setint_+)$ 
of sequences $u$ such that $\sup_{k\in\setint_+} |u_k| < \inpbound$.
The system operator 
$F:\ell_\infty(\setint_+)\to\setreal(\setint_+)$ is 
said to have fading memory on $\inpclass$ if there 
exists a decreasing sequence $w:\setint_+\to(0,1]$ with
$\lim_{k\to\infty} w_k = 0$ such that given an $\epsilon > 0$, 
there is a $\delta > 0$ such that
	$\max_{m\in\{0,1,\dotsc,k\}} 
	\| u_m - \tilde{u}_m \|_\infty w_{k-m} < \delta$
implies
	$\left| (Fu)_k - (F\tilde{u})_k\right| < \epsilon$
for every $u,\tilde{u} \in \inpclass$ and every 
$k\in\setint_+$. 
\end{defn}

In words, a system operator has fading memory if 
\textit{``two inputs which are close in the recent past 
(but not necessarily close in the remote past) yield 
outputs which are close in the present''} 
\cite{boyd_fading_1985}.
In particular, fading memory operators share many
properties with stable linear systems. An example of
this behavior can be seen in Figure 
\ref{fig:HH_full_Example} (left column): the response of 
$\Fint$ to a step input is reminiscent of that of a linear system, 
and its output tends to a unique steady-state value
whenever its input tends to a constant as $k\to \infty$. 

\begin{figure}[t]
	\centering
	\begin{minipage}[t]{0.45\linewidth}
		\if\usetikz1
		\begin{tikzpicture}
				\foreach \i/\j/\k in 
				{1/blue/solid,2/purple/solid,3/red/solid}{
				\begin{axis}[
					small,
					ylabel={$(\Fint v)_k$}, 
					axis y line = left,
					axis x line = bottom, 		
					ymin=-1200,ymax=1000,
					each nth point={10},
					xticklabels=\empty,
					ytick={-500,0,500},
					at={(0,{(-\i+3)*-1.50cm})}
					]			
					\addplot [color=\j,\k,thick]
						table[x index=0,y index=\i] 
					{./Data/HH_ions_stepExample_output.txt};	
				\end{axis}
			}
			\begin{axis}[small,
				ylabel={$v_k$}, 
				axis y line = left,
				xlabel={$k\samplingT\;\mathrm{[ms]}$}, 
				axis x line = bottom,
				ymax=-20,
				each nth point={10},
				ytick={-65,-45,-25},
				at={(0,3*-1.5cm)}
				]
				\foreach \i/\j/\k in 
					{1/blue/solid,2/purple/solid,3/red/solid}{
					\edef\temp{\noexpand
						\addplot [color=\j,\k,thick]
						table[x index=0,y index=\i] 
						{./Data/HH_ions_stepExample_input.txt};
					}\temp
				}
			\end{axis}
		\end{tikzpicture}
		\else
			\includegraphics[scale=1]{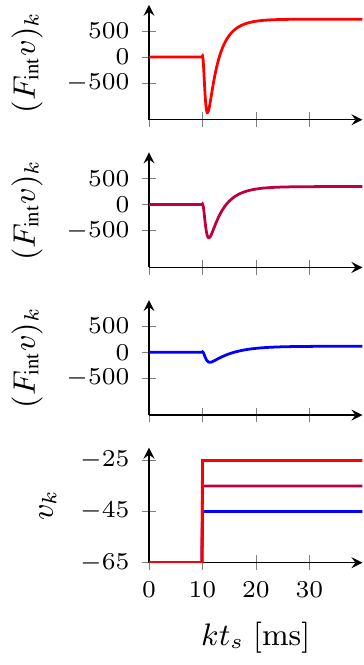}
		\fi		
	\end{minipage}\hfill
	\begin{minipage}[t]{0.45\linewidth}		
		\if\usetikz1
		\begin{tikzpicture}
				\foreach \i/\j/\k in 
				{1/blue/solid,2/purple/solid,3/red/solid}{
				\begin{axis}[small,
					ylabel={$v_k$}, 
					axis y line = left,
					axis x line = bottom, 		
					ymin=-80,ymax=60,
					each nth point={3},
					xticklabels=\empty,
					ytick={-65,0},
					at={(0,{(-\i+3)*-1.50cm})}
					]			
					\addplot [color=\j,\k,thick]
						table[x index=0,y index=\i] 
							{./Data/HH_full_stepExample_output.txt};	
	
				\end{axis}
			}
			\begin{axis}[small,
				ylabel={$\iApp_k$}, 
				axis y line = left,
				xlabel={$k\samplingT\;\mathrm{[ms]}$}, 
				axis x line = bottom,
				ymax=8,
				each nth point={10},
				ytick={2,4,6},
				at={(0,3*-1.5cm)}
				]
				\foreach \i/\j/\k in 
					{1/blue/solid,2/purple/solid,3/red/solid}{
					\edef\temp{\noexpand
						\addplot [color=\j,\k,thick]
						table[x index=0,y index=\i] 
						{./Data/HH_full_stepExample_input.txt};
					}\temp
				}
			\end{axis}
		\end{tikzpicture}
		\else
			\includegraphics[scale=1]{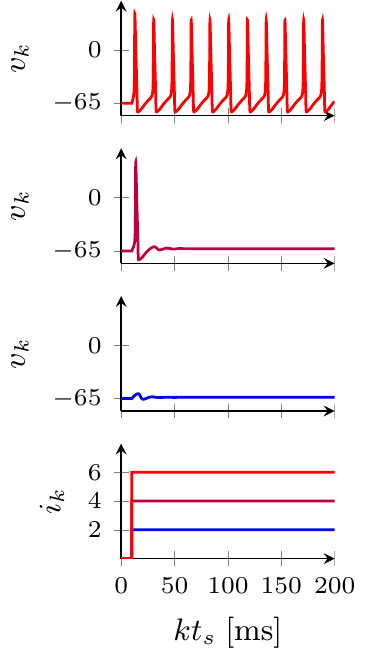}
		\fi	
	\end{minipage}
	\caption{
	Left column: Typical aggregate ionic current behavior
	in a biological neuron under voltage-clamp. Three
	different experiments are shown in which the
	reference $r_k$ is stepped to $-45$, $-35$ and
	$-25$ mV at $10$ ms. Right column: Typical behavior 
	of a neuron (without voltage-clamp) subject to steps 
	of different amplitudes at the input $i_k$. 
	Voltages are in $\mathrm{[mV]}$, and currents
	are in $\mathrm{[\upmu A/cm^2]}$.}	
	\label{fig:HH_full_Example}
\end{figure}
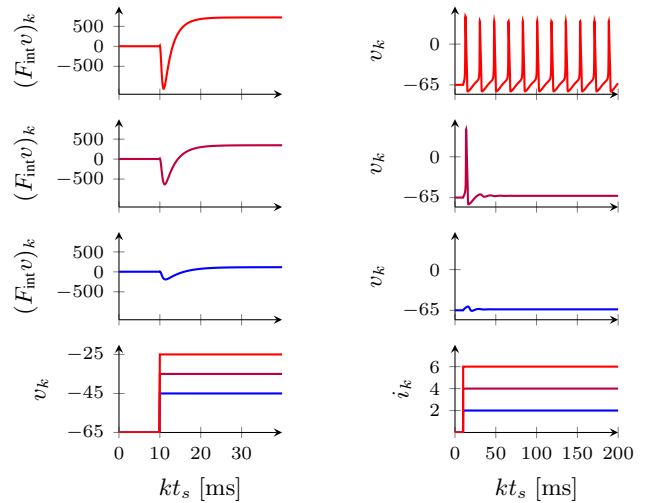	

While the fading memory of $\Fint$ can be assumed
from empirical evidence, it
can also be derived from continuous-time
biophysical conductance-based neuronal models. In 
\cite{burghi_feedback_2021}, it is shown that
all such models have an exponentially contracting
internal dynamics (as defined by
\cite{lohmiller_contraction_1998}), and that 
choosing a sufficiently small $\samplingT$ preserves
this property in the discretized system. Ultimately,
this implies that in conductance-based models, the
operator $\Fint$ has fading memory with an
exponentially fast fading rate.

\subsection{Identification problem}
\label{sec:statement}
 
We pose the following problem statement:
\textit{Identify the input-output behavior of the system
\eqref{eq:DT_v_system}-\eqref{eq:fm_channels}, where the
operator $\Fint$ has fading memory on $\inpclass$.}
In other words, we search for a model that accurately 
reproduces the set of admissible input-output trajectories 
$(i,v) \in \setreal(\setint_+)^2$ of the system being 
identified. 

In this paper, we focus on the question of 
model structure selection, and ignore the issue of
noise. Noise plays an important role in neuronal dynamics,
and input current (process) noise is particularly relevant 
\cite{gerstner_neuronal_2014}. However, as long as measurement 
noise is kept low, the methods discussed in the next
section should not be significantly affected (see 
\cite{burghi_feedback_2021} for related results where
input noise is taken into account).


\section{Identification method}
\label{sec:identification}

System operators that have fading memory can be uniformly
approximated  by simple classes of nonlinear systems,
called universal approximators. The parametrization of
fading memory operators has a long history going back to
Volterra series with notable contributions by Boyd and Chua
in \cite{boyd_fading_1985} and a specific link to
neuroscience (though in a different context from the
present paper) by Maas and Sontag \cite{maass_neural_2000}.
To tackle the problem stated in Section 
\ref{sec:statement}, we propose to use the ``direct approach''
of closed-loop identification \cite{forssell_closed-loop_1999}
and directly identify $\Fint$ with a model structure based on 
the series interconnection of Generalized Orthonormal Basis 
Functions (GOBFs)
\cite{ninness_unifying_1997}
and static Artificial Neural Networks (ANNs). This
structure was first proposed by \cite{schram_system_1996}.
In this section, we first define this structure and show 
that it has universal approximation power; then, we 
discuss how it can be tailored to identify neuronal models.

\subsection{A model structure for identifying $\Fint$}

\label{sec:model_structure} 
%
%
We define an ANN with logistic nonlinearities as
follows:

\begin{defn}
Let $\nlayers \in \setnat$ be the number of layers 
in the network. For $\ell = 1,\dotsc, \nlayers$, let 
$\nunits^{(\ell)}\in\setnat$ be the number of 
activation functions in the $\ell^\text{th}$ hidden 
layer, and $\nunits^{(0)}$ and 
$\nunits^{(\nlayers+1)}$ the number of inputs and
outputs of the network, respectively. Given the
weight matrices $W^{(\ell)} \in \setreal^{\nunits^{(\ell+1)}\times 
\nunits^{(\ell)}}$ and the bias vectors $b^{(\ell)} \in 
\setreal^{\nunits^{(\ell+1)}}$, let
\begin{equation}
	y^{(\ell)} = W^{(\ell)} u^{(\ell)} + b^{(\ell)}
\end{equation}
where $\ell=0,1,\dotsc,\nlayers$, and 
\begin{equation}
	\label{eq:act_functions_nn} 
		u_a^{(\ell)} = 
		\left(1+\text{exp}
		\left(-y_a^{(\ell-1)}\right)\right)^{-1}
\end{equation}
where $\ell=1,\dotsc,\nlayers$ and 
$a = 1,\dotsc,\nunits^{(\ell)}$. We define an artificial neural 
network nonlinearity $\psi: u \mapsto y$ by setting 
$u^{(0)} := u$ and $y := y^{(\nlayers)}$, so that
that $\nunits^{(0)} = \text{dim}(u)$ and
$\nunits^{(L+1)} = \text{dim}(y)$.
\end{defn}

Sandberg et al. showed that single-layer
time-delay neural networks constitute a universal
approximator for fading memory operators
\cite[Theorem 1]{sandberg_approximation_1991} and
\cite[Proposition 1]{park_criteria_1992}:

\begin{lem}
\label{lem:tdnn} 
Consider a time-invariant 
causal 
operator 
$F:\ell_\infty(\setint_+)\to\setreal(\setint_+)$ with 
fading memory on $\inpclass$. 
For $\itdnn \in \setnat$, let 
$H^{(\itdnn)}:\ell_\infty(\setint_+)\to
\setreal(\setint_+)^{(\itdnn+1)}$ be given by
\begin{equation}
	\label{eq:tdnn} 
	(H^{(\itdnn)} v)_k = [v_k,v_{k-1},\dotsc,v_{k-\itdnn}]^\top,
\end{equation}
and let $\psi$ be a neural network with $\nlayers = 1$,
$\nunits^{(0)} = \itdnn+1$ and $\nunits^{(2)} =
1$. Then, there are $\itdnn,\nunits^{(1)}>0$ and
real parameters $W^{(0)}, W^{(1)}, b^{(0)},b^{(1)}$
such that for any $\epsilon > 0$,
\[
	|(F v)_k - \psi((H^{(\itdnn)} v)_k)| < \epsilon,
	\quad\quad k\in\setint_+
\]
for all $v \in \inpclass$.
\end{lem}

%

To improve on the time-delay model structure above,
we follow \cite{schram_system_1996} and replace 
\eqref{eq:tdnn} by a set of GOBFs:
\begin{defn}(see \cite{ninness_unifying_1997})
Let
$\{\xi_0,\xi_1,\xi_2,\dotsc\}$ be a sequence of
(possibly complex) poles such that 
$|\xi_i|<1$ for all $i \in \setint_+$. The set of 
Generalized Orthonormal Basis Functions  is 
defined by the causal transfer functions
\begin{equation}
	\label{eq:gobfs} 
	\begin{split}
	G_0(z) &= z^d \frac{\sqrt{1-|\xi_0|^2}}{z-\xi_0} \\
	G_i(z) &= 
	z^d \frac{\sqrt{1-|\xi_i|^2}}{z-\xi_i} 
	\prod_{j=0}^{i-1} 
	\frac{1-\bar{\xi}_j z}{z-\xi_j}, \quad i = 1,2,\dotsc
	\end{split}
\end{equation}
with $d=0$ or $d=1$.
\end{defn}

Note that when $d=1$ and all the poles in 
$\{\xi_i\}_{i\in\setint_+}$ are equal to zero, we recover
the set of time-delay basis functions. 
The utility of GOBFs comes from
the fact that poles are allowed to be distinct, and 
thus the choice of poles may be adapted to the system which
is being approximated. 
The set of GOBFs may form a basis for various system
spaces. In particular, a simple condition ensures that it
is fundamental\footnote{A fundamental set in a normed
space $X$ is a subset $M \subset X $ whose span is dense in
$X$ \cite{kreyszig_introductory_1989}.} in 
$\ell_1(\setint_+)$, the space of causal discrete-time LTI 
systems whose impulse responses are absolutely integrable 
\cite[Corollary 10] {akcay_rational_1998}:

\begin{lem}
	\label{lem:fundamental_l1}
	The set \eqref{eq:gobfs}, with $d=1$, is fundamental in
	$\ell_1(\setint_+)$ if $k^{-\alpha}=O(1-|\xi_k|)$ for some 
	$0<\alpha<1$.
\end{lem}


The next result, which, as far as we know, has not been 
previously published, generalizes Lemma \ref{lem:tdnn}:

\begin{thm}
\label{thm:gobfnn}
Consider a time-invariant 
causal 
operator 
$F:\ell_\infty(\setint_+)\to\setreal(\setint_+)$ with 
fading memory on $\inpclass$. Let 
$\{\xi_i\}_{i\in\setint_+}$ be a sequence of poles 
satisfying the condition stated in Lemma
\ref{lem:fundamental_l1}. For $\igobf \in \setnat$, let
$G^{(\igobf)}:\ell_\infty(\setint_+)
	\to\setreal(\setint_+)^{(\igobf+1)}$
be defined by
\begin{equation}
	\label{eq:gobfnn} 
	(G^{(\igobf)} v)_k = [(G_0 v)_k, (G_1 v)_k,\dotsc,
	(G_\igobf v)_k]^\top
\end{equation}
with the $G_i$ being the operators associated to
\eqref{eq:gobfs} with $d=1$. Let $\psi$ be a neural network
with $\nlayers \ge 1$, $\nunits^{(0)} = \igobf+1$ and 
$\nunits^{(\nlayers+1)} = 1$. Then, there are  
$\igobf,\nunits^{(\ell)}>0$ and real parameters
$W^{(\ell)}, b^{(\ell)}$ ($\ell = 1,\dotsc,\nlayers$)
such that for any $\epsilon > 0$,
\begin{equation}
	\label{eq:approx_nn_gobfs} 
	|(F v)_k - \psi((G^{(\igobf)} v)_k)| < \epsilon,
	\quad\quad k\in\setint_+
\end{equation}
for all $v \in \inpclass$.
\end{thm}
\begin{proof}
	See Appendix \ref{proof:gobfnn}.
\end{proof}



To obtain a model structure that satisfies Lemma 
\ref{lem:fundamental_l1} and allows for a direct
term from input to output, we fix a finite sequence 
$\{\lambda_1,\lambda_2,\dotsc,\lambda_{\npoles}\}$, 
a number $n_{\text{rep}} \in \setnat$, and set 
$\xi_0 = 0$ and
		$\xi_{i+j \npoles} = \lambda_i$,
		for $i = 1,\dotsc,\npoles$ and
		$j = 0,\dotsc,n_{\text{rep}}-1$.    
	Let,
	\begin{equation*} 
	G(z) = [G_0(z), G_1(z),\dotsc,
	G_{\npoles n_{\text{rep}}}(z)]^\top
	\end{equation*}
	where the SISO transfer functions $G_i(z)$ are given
	by \eqref{eq:gobfs}, with $d=1$. Let 
	$\psi(\,\cdot\,;\theta)$ be an ANN specified by 
	$\nlayers$ and $\nunits$, with the vector of parameters
	$\theta$ encompassing the elements of the matrices
	$W^{(\ell)}$ and the elements of the bias vectors
	$b^{(\ell)}$. Then the model structure is given by 
	\begin{equation}
		\label{eq:model_structure} 
		\begin{split}
			\hat{y}_k(\theta) = \psi(G(q) v_k;\theta) 
		\end{split}
	\end{equation}
	where $q$ is the forward shift operator.
	
\subsection{GOBFs and spiking signals}

The set of GOBFs get their name from the fact that, under
a mild condition on the sequence of poles 
$\{\xi_i\}_{i\in\setint_+}$, it also forms a 
fundamental \textit{orthonormal} set in 
$\mathcal{H}_2({\mathbb{C}_+})$, 
the space of functions which are analytic on 
$\mathbb{C}_+=\{z\;:\;|z|>1\}$ and square integrable on the 
unit circle \cite{ninness_unifying_1997}. 
In particular, we have 
\[
	\sum_{k=0}^\infty (g^*_i)_k (g_j)_k = 0, \quad 
	\forall \, i\neq j, \quad i,j \ge 0
\]
where here the asterisk denotes complex conjugation, and
$(g_i)_k$ are GOBF impulse responses at time $k$.

This property is relevant in the context of neuronal system
identification. Figure \ref{fig:spike_gobfs} illustrates that 
for GOBFs defined with timescales which are much larger than
those of a hypothetical input spike with voltage trace $v_s$, the 
responses $g_i * v_s $ to the spike are similar to the 
(mutually orthogonal) impulse responses $g_i$. This means that, 
for GOBFs with a relatively slow timescale, a fast input spike 
behaves approximately as an impulse, and thus the outputs of different 
slow-timescale GOBFs to the same fast spike will be close to orthogonal. 
In the model structure \eqref{eq:model_structure}, approximate 
orthogonality of the signals at the output of the GOBFs ensures that the ANN 
nonlinearity receives a rich set of inputs to operate on.

\begin{figure}
	\centering
	\if\usetikz1
	\begin{tikzpicture}
		\begin{axis}[height=4.5cm,
			width=8cm,
			axis on top,
			axis y line = left,
			axis x line = bottom,
			ylabel={$(v_s)_k \;\mathrm{[mV]}$}, 
			xlabel={$kt_s\;\mathrm{[ms]}$}, 
			xmax=100,
			ymin=-0.1,
			each nth point={3},
			extra x ticks = {6.65},
			extra x tick label = {$t_1$},
			name=v_gobf
			]			
			\addplot [color=black,thick]
				table[x index=0,y index=1] 
					{./Data/gobf_dt_spike.txt};	

			\end{axis}
		\begin{axis}[height=4.5cm,
			width=8cm,
			axis on top,
			axis y line = left,
			axis x line = bottom,
			ylabel={$(g_i*v_s)_k$}, 
			xlabel={$k t_s\;\mathrm{[ms]}$}, 
			xmax=100,
			ymin=-2 ,ymax=6,
			each nth point={3},
			anchor = north,
			yshift = -2cm,	
			legend entries={$i=0$,$i=1$,$i=2$,$i=3$},
			legend pos =  north east,
			at={(v_gobf.south)},
			name=y_gobf
			]
					
			\foreach \i/\j/\k in 
				{2/blue/solid,3/green/solid,
					4/red/solid,5/cyan/solid}{
				\edef\temp{\noexpand
					\addplot [color=\j,\k,thick]
					table[x index=0,y index=\i] 
					{./Data/gobf_dt_spike.txt};
				}\temp
			}
					
		\end{axis}
		\begin{axis}[height=4.5cm,
			width=8cm,
			axis on top,
			axis y line = left,
			axis x line = bottom,
			ylabel={$(g_i)_k$}, 
			xlabel={$kt_s\;\mathrm{[ms]}$}, 
			xmax=100,
			ymin=-2,ymax=6,
			each nth point={3},
			anchor = north,
			yshift = -2cm,
			legend entries={$i=0$,$i=1$,$i=2$,$i=3$},
			legend pos =  north east, 
			name=h_gobf,		
			at={(y_gobf.south)}
			]
					
			\foreach \i/\j/\k in 
				{2/blue/solid,3/green/solid,
					4/red/solid,5/cyan/solid}{
				\edef\temp{\noexpand
					\addplot [color=\j,\k,thick]
					table[x index=0,y index=\i] 
					{./Data/gobf_dt_impulse.txt};
				}\temp
			}
					
		\end{axis}
	\end{tikzpicture}
	\else
		\includegraphics[scale=1]{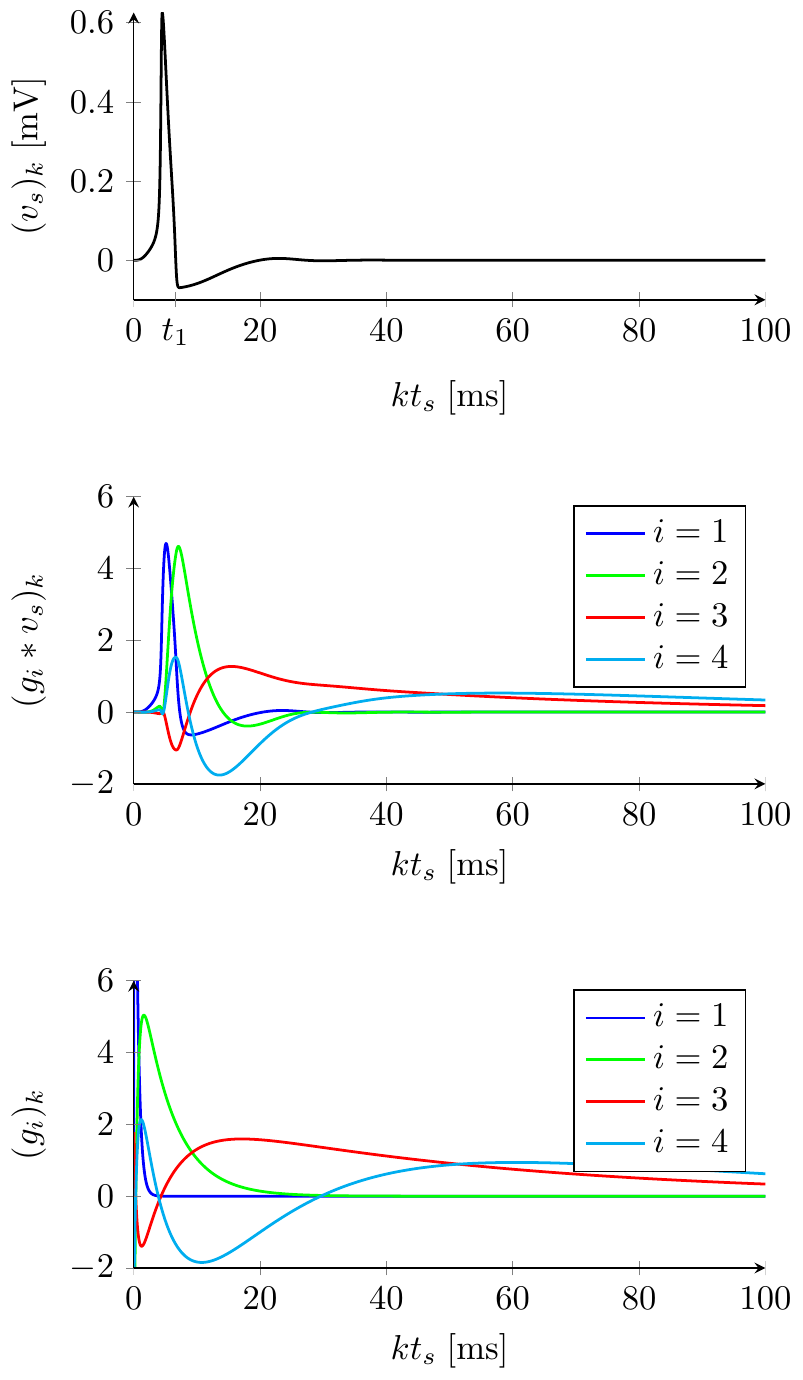}
	\fi
	\caption{Comparison between impulse responses (bottom) and 
	spike responses (middle) of four GOBFs given by
	\eqref{eq:gobfs} with $d=0$, $\xi_0 = 0.9802$, $\xi_1= 0.9980$, 
	$\xi_2= 0.9998$ and $\xi_3= 0.9995$.  The spike (top) used to
	compute the responses $g_i * v$ was normalized so that the
	area between $t=0$ and $t=t_1$ equals one (the same area
	of a Dirac impulse). The impulse responses $g_i$ are mutually 
	orthogonal, while the responses	$g_i * v_s$ are not. However,
	the responses  $g_3*v_s$ and $g_4*v_s$ are close to $g_3$ and
	$g_4$, respectively, and thus are close	to being orthogonal
	to each other. This can be explained by the	fact that the
	timescales of $g_3$ and $g_4$ are much larger than those
	contained in the spike. The sampling time used is 
	$t_s = 0.01$ ms.
	}	
	\label{fig:spike_gobfs}
\end{figure}

\subsection{GOBF pole selection for neuronal 
identification}
\label{sec:pole_selection} 

Since the choice of GOBF poles $\{\xi_i\}$
in the model structure is arbitrary, we pose the
question of how to choose them in order to minimize the
number of basis functions required to approximate 
\eqref{eq:fm_channels} to a given
tolerance. If the operator $\Fint$ were linear and finite
dimensional, with a transfer function $\Fint(s)$, then the 
obvious choice would be to place $\{\xi_i\}$ as close 
as possible to the poles of $\Fint(s)$. 
However, since $\Fint$ is nonlinear and, in
general, infinite-dimensional, it is not immediately clear
how to choose $\{\xi_i\}$. 
With that in mind, we now consider a heuristic for choosing
the GOBF poles, keeping the discussion at the conceptual
level. 

The heuristic involves choosing the GOBF poles so
as to ensure that the linearized model structure 
\[
	\delta\hat{y}_k(\theta) = 
	\left.\frac{\partial \psi(u;\theta)}{\partial u}
	\right|_{u=G(1)\bar{v}}
	G(q) \delta v_k
\]
(whose poles are $\{\xi_i\}$) is able to approximate
well the small-signal behavior of $\Fint$ around a particular
setpoint $\bar{v}\in\setreal$. This is the behavior of 
$\Fint$ subject to an input $v_k = \bar{v} + \tilde{v}_k$, 
$k\ge 0$, where $|\tilde{v}_k|$ has a small amplitude. 
To identify excitable neuronal systems, the choice of 
the setpoint $\bar{v}$ can be informed by the 
closed-loop nature of 
\eqref{eq:DT_v_system}-\eqref{eq:fm_channels}.

We argue that this setpoint $\bar{v}$ should be a point of
\textit{ultra-sensitivity} \cite{franci_sensitivity_2019} 
of the system \eqref{eq:DT_v_system}-\eqref{eq:fm_channels}.  
To explain this feature of neuronal behaviors, we consider 
a case in which the fading memory operator $\Fint$ has a smooth
finite-dimensional realization. Since $\Fint$ has fading 
memory, when it is subject to a constant input 
$v_k = \bar{v}$, $k\ge 0$, its output converges to 
a unique constant value $i_\infty\in\setreal$, that is,
$(\Fint v)_k \to i_\infty$ as $k\to\infty$ (see, e.g.,
\cite{boyd_fading_1985}). 
Let $d\Fint(z;\bar{v})$ be the transfer function 
associated to the linearized realization of $\Fint$ 
at $\bar{v}$.
Because $\bar{v}$ is also an equilibrium of the 
closed-loop system 
\eqref{eq:DT_v_system}-\eqref{eq:fm_channels} for
$i_k = i_\infty$, we can write down the transfer function
of the linearized closed-loop system at $\bar{v}$ and
$i_\infty$ as 
\begin{equation}
	\label{eq:closed_loop_tf} 
	T(e^{j\omega\samplingT};\bar{v}) = 
	\frac{\samplingT}{c (e^{j\omega\samplingT} - 1)+
		\samplingT \, 
		d\Fint(e^{j\omega \samplingT};\bar{v})}
\end{equation}
where $\omega \ge 0$ denotes the frequency variable.
Neuronal systems are characterized by the existence of at least 
one point $(\bar{\omega},\bar{v})$ where \eqref{eq:closed_loop_tf} 
is singular. At such a point, $T(z;\bar{v})$ has a pole on the 
unit circle, which, in terms of state-space dynamics, corresponds 
to an equilibrium bifurcation at $\bar{v}$ (see 
\cite{izhikevich_dynamical_2007} for a state-space 
perspective). 
For system identification, the most relevant 
bifurcation occurs at the equilibrium voltage $\bar{v}=v^*$ 
where the neuronal behavior transitions from a constant-steady 
state to an oscillatory or chaotic motion. We call $v^*$ a
point of ultra-sensitivity \cite{franci_sensitivity_2019}.

Importantly, the point of ultra-sensitivity can be characterized
experimentally; this can be done, for instance, by stepping 
the input $i_k$ to different constant values, as illustrated
in Figure \ref{fig:HH_full_Example}.
The critical role that $v^*$ plays in the dynamics of
\eqref{eq:DT_v_system}-\eqref{eq:fm_channels} suggests that 
we should place the GOBF poles as close as possible to 
the (stable) poles of $d\Fint(z;v^*)$, since those are 
the poles that shape the closed-loop frequency-response
\eqref{eq:closed_loop_tf} at a point of ultra-sensitivity. 

This heuristic requires that the poles of 
$d\Fint(z;v^*)$ be identified prior to identifying $\Fint$. 
But since small changes in the current $i_k$ may cause large 
voltage deviations away from $v^*$, it is not immediately 
clear how to achieve that. The answer lies in 
\textit{voltage-clamp}: the output feedback law 
\eqref{eq:output_feedback} can be used to stabilize the system 
at the voltage $v^*$, suppressing any oscillations and 
allowing the use of a low-amplitude input $r_k$ to probe
the small-signal behavior of $\Fint$ around $v^*$. If signal
amplitudes are kept low, linear system identification methods
can be applied to the measurements obtained under the
output feedback law \eqref{eq:output_feedback}, leading to 
estimates for the poles of $d\Fint(z;v^*)$; these are in turn
used to define the GOBFs. 

This heuristic leads to a systematic choice for the GOBF poles 
that takes into account the ultra-sensitivity of neuronal systems. 
Notice that, in principle, a finite-dimensional $\Fint$ is 
not required.
The heuristic is in line with the idea that a pre-processing 
linear identification step should be implemented, if possible, in 
order to make an informed choice on GOBF poles for nonlinear 
system identification \cite{tiels_wiener_2014}. 



\subsection{Parameter estimation and validation}

In practice, given measurements of $v_k$ and the input 
$i_k$, the system 
\eqref{eq:DT_v_system}-\eqref{eq:fm_channels} can be
identified by solving 
\begin{equation}
	\label{eq:optimization_problem} 
	\min_{\eta,\theta} \frac{1}{N}
	\sum_{k=0}^N \left(\frac{v_{k+1}-v_k}{\samplingT}
	- (\hat{y}_k(\theta) + \eta \, i_k) \right)^2	
\end{equation}
where $\eta$ is used to obtain an estimate of $1/c$.
Because $\hat{y}_k(\theta)$ is a feed-forward ANN,
the above problem can be efficiently solved with
backpropagation algorithms. Once this is achieved,
we can define
\begin{equation}
	\label{eq:cl_model}
	\hat{v}_{k+1} = \hat{v}_k + t_s \eta
	(-\psi(G(q)\hat{v}_k;\theta) + i_k)
\end{equation}
to obtain a closed-loop identified neuronal model.

Due to the frequent bifurcations happening in a 
neuronal model, different metrics are used for
validating the subthreshold and the superthreshold
(i.e., spiking) dynamics of the model 
\cite[Section 10.3]{gerstner_neuronal_2014}. 
Here, we will focus on validating the latter. 
This can be done by comparing the spike timings of
the identified model, denoted by $\spm$, with the
spike timings of the validation dataset, denoted by
$\spd$ (we assume the timing of a spike is given by
the location of its maximum). Spike timings define
impulse trains (called spike trains) given by $s_k = \sum_{\spd} 
\delta_{k\spd}$ and $\hat{s}_k = \sum_{\spm}
\delta_{k\spm}$, where $\delta_{ij}=1$ for $i=j$ and
$\delta_{ij}=0$ otherwise.  
A measure of spike coincidence can be defined using 
the inner product of smoothed spike trains,
\[
	(s,\hat{s}) = \sum_{k=0}^N (w*s)(k) (w*\hat{s}_k)(k),
\]
where $w_k$ is a smoothing kernel. Here, we will
use a Gaussian kernel 
$w_k = \exp(-k^2/(2\rho^2))/\sqrt{2\pi\rho^2}$. The
angular separation between the smoothed spike
trains, measured by
\begin{equation}
	\label{eq:spike_coincidence} 
	\Delta_\rho = \frac{(s,\hat{s})}{\sqrt{(s,s)}\sqrt{(\hat{s},\hat{s})}},
\end{equation}
is a good measure of spike coincidence: 
$\Delta_\rho$ approaches $1$ for very similar 
spike trains, and $0$ for very different ones.

\section{Example}
\label{sec:example} 

The purpose of this section is to illustrate the impact 
of the choice of basis functions \eqref{eq:gobfs} on the 
identification of a single-compartment neuron. 
For this illustration, we use a neuronal model 
from the crab stomatogastric ganglion (STG), a system 
responsible for producing rhythmic muscular activity in the 
crab's stomach \cite{marder_understanding_2007}. Neurons 
in the STG are capable of bursting autonomously, and, to 
our knowledge, they have never been successfully identified 
using a generic model structure. To generate data for 
parameter estimation, we use a forward-Euler discretization
of the conductance-based STG neuron model described in 
\cite[Figure 1.A.a]{franci_robust_2017}; see also
\cite[pp. 2318-2319]{liu_model_1998} for the model's 
ion channel kinetics. 
The model in question is a state-space model with 
12 states representing the membrane voltage and 
six different types of ionic currents. It contains 
over a hundred parameters, which were originally 
determined by first using ad-hoc methods for fitting 
individual ionic currents, and then hand tuning the 
parameters to match the observed neuronal behavior.

We compare the validation performance of the closed-loop
model \eqref{eq:cl_model} in two cases: when the basis 
functions used to define $G(z)$ are given 
by time-delays, and when they are given by GOBFs with 
poles chosen in accordance with Section \ref{sec:pole_selection}.
We denote the basis function vector in each of these
cases by $G^{\mathrm{td}}(z)$ and $G^{\mathrm{gobf}}(z)$,
respectively.
In both cases, we use the same two-layer ANN parametrization, 
with $\nunits^{(1)} = 15$ and $\nunits^{(2)} = 12$.
Letting $n_{\mathrm{td}}$ be a maximum time delay, we set 
$G^{\mathrm{td}}(z) = [1, z^{-1},\dotsc,z^{-(n_{\mathrm{td}}-1)}]^\top$.
To choose the GOBF poles used to define $G^{\mathrm{gobf}}(z)$, 
for brevity, we assume knowledge of the local dynamics of 
the continuous-time STG model\footnote{For numerical 
simulations where the poles of $d\Fint(z;\bar{v})$ are 
identified using output feedback and linear identification methods, see 
\cite[Section 4.6]{burghi_feedback_2020}.}. It is known that 
the STG model starts bursting due to a saddle-node bifurcation 
close to $i_\infty = -0.25$ and $\bar{v} = -49 \; \mathrm{mV}$. 
At this equilibrium, the linearized $\Fint$ has 
poles given by $\{\xi^{\mathrm{ct}}_i\}_{i=1,\dotsc,11} = \{-7.716$, 
$-0.560$, $-0.179$,$-0.156$, $-0.112$, $-0.074$, $-0.050$, 
$-0.037$, $-0.019$,$-0.018$, $-0.004\}$. We define twelve
GOBFs ($d=1$) by setting discrete-time
poles at $\xi_0 = 0$ and $\xi_i = 1+\samplingT\xi^{\mathrm{ct}}_i$, 
with $\samplingT$ the sampling time (we use $n_{\text{rep}}=1$).

Using $\samplingT=0.0075$, we indentify the excitability 
of the discretized STG model based on its response $v_k$ to the input
\begin{equation}
	\label{eq:excitation} 
	\iApp_k = -0.5 + \varepsilon_k
\end{equation}
where $\varepsilon_k$ is a Gaussian white-noise signal
such that $\sigma[\varepsilon_k] = 20$ $\mathrm{\upmu A/cm^2}$.
We obtain a training and a validation dataset, both
of length $N\samplingT\approx 10\,\mathrm{s}$, using
two different realizations of \eqref{eq:excitation}.
After training\footnote{We trained the ANNs of the
two model structures  by solving 
\eqref{eq:optimization_problem} with the Levenberg-
Marquadt backpropagation algorithm. The models were
trained with ten different sets of randomly selected
initial parameters, and the best fitted models were
used in the results. To eliminate transient effects,
we discarded $1\,\mathrm{s}$ of data from the 
training dataset.} each of the models defined with
$G^{\mathrm{td}}$ and $G^{\mathrm{gobf}}$, we
simulated the resulting system \eqref{eq:cl_model}
with the validation input, obtaining  
$\hat{v}^{\mathrm{td}}_k$ and
$\hat{v}^{\mathrm{gobf}}_k$, respectively. 

Figure 
\ref{fig:STG_coloured_nrep_1_na1_15_na2_12_validation}
shows the validation voltage $v_k$ and the simulated
voltages $\hat{v}^{\mathrm{td}}_k$ and 
$\hat{v}^{\mathrm{gobf}}_k$ during; we have used 
$n_\mathrm{td} = 12$ so that both models have the
same number of basis functions. The spike
coincidence metric 
\eqref{eq:spike_coincidence}, computed with
$\rho=3$ using the full validation 
dataset\footnote{The choice of $\rho=3\,\mathrm{ms}$ 
for the smoothing kernel standard deviation was
based on the mean width of the observed spikes,
which was of $6\,\mathrm{ms}$.}, was 
$\Delta_3 = 0.42$ for the model with time-delays, 
and $\Delta_3 = 0.73$ for the model with GOBFs. 
Thus it can be said regarding the model structure
\eqref{eq:model_structure} that using GOBFs with 
judiciously chosen poles leads to better results
than using time delay basis functions. In 
particular, the model obtained with a maximum delay
of $n_\mathrm{td} = 12$ samples can \textit{spike},
but it cannot \textit{burst}. This highlights the
importance of long timescales in bursting models.

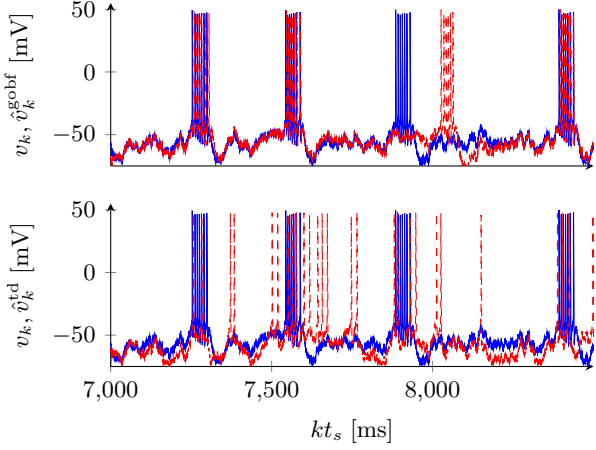
\begin{figure}
	\centering
	\if\usetikz1
	\begin{tikzpicture}
		\begin{axis}[height=3.75cm,width=8cm,
		ylabel={$v_k,\hat{v}_k^{\mathrm{gobf}} \; \mathrm{[mV]}$}, 
		axis y line = left,
		axis x line = bottom,
		ymin=-75,ymax=55,
		legend pos=outer north east,
		each nth point={15},
		xtick=\empty,
		tick label style={font=\small},
		label style={font=\small},
		name = v_gobf
		]
		
		\addplot [color=blue,opacity=1]
		table[x index=0,y index=1]{./Data/GOBF_validation.txt};
		\addplot [color=red,opacity=1,dashed]
		table[x index=0,y index=2]{./Data/GOBF_validation.txt};

	\end{axis}
	\begin{axis}[height=3.75cm,width=8cm,
		ylabel={$v_k,\hat{v}_k^{\mathrm{td}} \; \mathrm{[mV]}$}, 
		axis y line = left,
		xlabel={$k\samplingT \; \mathrm{[ms]}$}, 
		axis x line = bottom,
		xtick={7000,7500,8000},
		ymin=-75,ymax=55,
		tick label style={font=\small},
		label style={font=\small},
		legend pos=outer north east,
		each nth point={15},
		anchor = north west,
		at={(v_gobf.south west)},
		yshift = -0.5cm
		]
		\addplot [color=blue,opacity=1]
		table[x index=0,y index=1]{./Data/TDNN_validation.txt};
		\addplot [color=red,opacity=1,dashed]
		table[x index=0,y index=2]{./Data/TDNN_validation.txt};
	\end{axis}
	\end{tikzpicture} 	
	\else
		\includegraphics[scale=1]{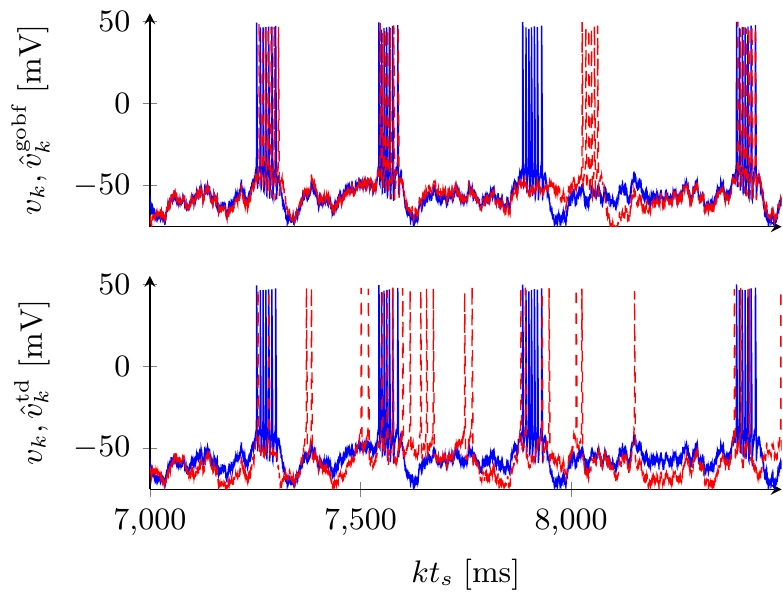}
	\fi
	\caption{Top: validation  $v_k$ of the true 
	STG model (blue, solid) and output 
	$\hat{v}_k^{\mathrm{gobf}}$ of the model defined
	with twelve GOBFs (red, dashed). 
	Bottom: 	 validation $v_k$ of the true 
	STG model (blue, solid) and output 
	$\hat{v}_k^{\mathrm{td}}$ of the model defined
	with $n_\mathrm{td} = 12$ time delays 
	(red, dashed). }		
\label{fig:STG_coloured_nrep_1_na1_15_na2_12_validation}
\end{figure}

\section{Conclusion}


In this paper, we have shown that GOBFs and static
ANNs can be used to
efficiently solve the problem of identifying
biophysical neuronal systems. One of the main
advantages of the method is its simplicity:
the estimation of model parameters can be done with 
off-the-shelf backpropagation algorithms. 
In addition, the method does not rely on measurements of
internal states of the neuron -- something which is,
in practice, impossible to obtain. 

The method also dispenses entirely with the conductance-based 
formalism, in which optimization of all model parameters 
may become an intractable problem. This comes, of 
course, at the cost of biophysical interpretability. 
In fact, further work on this topic should aim to 
understand how much the model structure discussed here can
be improved so as to provide more interpretability while
retaining its capabilities as a universal approximator.
 
Finally, we point out that while the use of feedforward 
ANNs for the identification of nonlinear systems 
is classical, such networks can only be guaranteed to 
approximate fading memory operators, de facto excluding 
the excitable nature of neuronal systems. Our approach 
shows that by retaining the biophysical feedback structure 
of the model in the identification step, such a limitation
can be overcome. 

\appendix
\section{Proofs}
\subsection{Proof of Theorem \ref{thm:gobfnn}}
\label{proof:gobfnn} 
We first prove the result for a single-layer ANN 
($\nlayers = 1$) by modifying the input	weights 
of the ANN obtained from Lemma \ref{lem:tdnn}. 
Let $\epsilon_0 > 0$ be given. By Lemma \ref{lem:tdnn},
there is a number $\itdnn=\itdnn(\epsilon_0)>0$, and
parameters $\nunits^{(1)}>0$ 
and $W_{\text{fir}}^{(0)}~\in~\setreal^{\nunits^{(1)}
\times (\itdnn+1)}$, defining an artificial neural
network $\psi:\setreal^{(\itdnn+1)}\to\setreal$ 
such that
\[
	|(F v)_k - \psi((H^{(\itdnn)} v)_k)| < \epsilon_0,
	\quad\quad k\in\setint_+
\]
for all $v \in \inpclass$.
Let $W_{\text{obf}}^{(0)} \in 
\setreal^{\nunits^{(1)}\times (\igobf+1)}$, and consider
the vector of transfer functions
\[
	Q^{(\itdnn,\igobf)}(z) = 
	W_{\text{fir}}^{(0)}H^{(\itdnn)}(z) - 
	W_{\text{obf}}^{(0)}G^{(\igobf)}(z) 
\]	
By Lemma \ref{lem:fundamental_l1}, for every $\itdnn$
and every $\epsilon_1>0$, there exists an 
$\igobf=\igobf(\itdnn,\epsilon_1)>0$, and a certain
$W_{\text{obf}}^{(0)}$ such that
\[
	\|q_i^{(\itdnn,\igobf)}\|_{\ell_1} 
	= \sum_{k=0}^\infty |(q_i^{(\itdnn,\igobf)})_k|		
	\le 
	\epsilon_1, \quad\quad i=1,\dotsc,\nunits^{(1)},
\]
where 
$q_i^{(\itdnn,\igobf)}$ is the impulse response of the
$i^{\text{th}}$ element	of $Q^{(\itdnn,\igobf)}(z)$.
Now, let $\tilde{\psi}:\setreal^{\nunits^{(1)}}\to\setreal$ 
be defined by $\tilde{\psi}(W_{\text{fir}}^{(0)}u) = \psi(u)$.
Since the activation functions
\eqref{eq:act_functions_nn} are globally Lipschitz,
there exists a constant $l_0$ such that
\begin{equation*}
	\begin{split}
		|\tilde{\psi}(W_{\text{fir}}^{(0)}&(H^{(\itdnn)}v)_k) -
				\tilde{\psi}(W_{\text{obf}}^{(0)}(G^{(\igobf)}v)_k)|\\ 
		&\le l_0 
		\|W_{\text{fir}}^{(0)}(H^{(\itdnn)} v)_k
		 -W_{\text{obf}}^{(0)}(G^{(\igobf)} v)_k\|_\infty \\
		&= l_0 \|( Q^{(\itdnn,\igobf)} v)_k\|_\infty \\
		&\le l_0 \max_i 
		\|q_i^{(\itdnn,\igobf)}\|_{\ell_1} 
			\|v\|_{\ell_\infty}
			\le \, l_0 \, \epsilon_1 \inpbound				
	\end{split}
\end{equation*}
for all $k\in\setint_+$. Here, we have used the fact
that, for an arbitrary LTI system 
$Q:\ell_\infty \to \ell_\infty^{\nunits}$, 
$\max_{i=1,\dotsc,\nunits} \|q_i\|_{\ell_1}$ is the induced
system gain \cite{zhou_robust_1996}, and that the
elements of $\inpclass$ are bounded by 
$\inpbound$.
From the above inequalities, we obtain
\begin{equation*}
	\begin{split}
	|(F v)_k - \tilde{\psi}(W_{\text{obf}}^{(0)}(G^{(\igobf)}v)_k)| 
	&\le |(F v)_k - \psi((H^{(\itdnn)}v)_k)| \\
		+ |\psi((H^{(\itdnn)}&v)_k)
			- \tilde{\psi}(W_{\text{obf}}^{(0)}(G^{(\igobf)}v)_k)| 	\\
	&\le \epsilon_0 + l_0 \, \epsilon_1 \inpbound		
	\end{split}
\end{equation*}
for all $k\in\setint_+$ and all $v \in \inpclass$. Thus 
given $\epsilon>0$, we can choose $\epsilon_0 < \epsilon/2$ 
and $\epsilon_1 < \epsilon/(2 l_0 \inpbound)$, and there
exists a large enough $\igobf$ such that the ANN 
$\tilde{\psi}(W_{\text{obf}}^{(0)} \,\cdot\,)$ 
satisfies \eqref{eq:approx_nn_gobfs}. 

This proves the result for $L=1$. The fact that the 
result also holds for a multi-layer ANN ($L>1$) follows 
from the fact that
multi-layer ANNs with continuous activation
functions are capable of arbitrarily accurate
approximation (in the uniform norm) to any continuous
function over a	compact set 
\cite[Theorem 2.1]{hornik_multilayer_1989}.
A multi-layer ANN can thus be used to approximate
the single-layer network obtained above to
arbitrary precision, concluding the proof.

\bibliographystyle{plain}
\bibliography{references}

\end{document}